\documentclass[twocolumn,showpacs,preprintnumbers, nofootinbib,superscriptaddress, amsmath,amssymb]{revtex4}
\usepackage{amsthm}
\usepackage{bm}
\usepackage{epsfig}

\DeclareFontFamily{OT1}{rsfs}{}
\DeclareFontShape{OT1}{rsfs}{m}{n}{<5> rsfs5 <7> rsfs7 <10> rsfs10}{}
\DeclareSymbolFont{mathrsfs}{OT1}{rsfs}{m}{n}
\DeclareSymbolFontAlphabet{\mathrsfs}{mathrsfs}

\newtheorem{theorem}{Theorem}
\newtheorem*{lemma}{Lemma}

\begin{document}

\title{The reproduction of the dynamics of a quantum system by an ensemble of classical particles beyond de Broglie--Bohmian mechanics} 

\author{Denys I. Bondar}
\email{dbondar@sciborg.uwaterloo.ca}
\affiliation{University of Waterloo, Waterloo, Ontario, Canada N2L 3G1}

\begin{abstract}
It is shown that for any given quantum system evolving unitarily with the Hamiltonian, $\hat{H} = \hat{\bf p}^2/(2m) + U({\bf q})$, [bold letters denote $D$-dimensional ($D \geqslant 3$) vectors] and with a sufficiently smooth potential $U({\bf q})$, there exits a classical ensemble with the Hamilton function,
$\mathrsfs{H} ({\bf p}, {\bf q}) = {\bf p}^2/(2m) + U^{(\infty)} ({\bf q})$,
where the potential $U^{(\infty)}({\bf q})$ coincides with $U({\bf q})$ for almost all ${\bf q}$ (i.e.,  $U^{(\infty)}$ can be different from $U$ only on a measure zero set), such that the square modulus of the wave function in the coordinate (momentum) representation approximately equals the coordinate (momentum) distribution of the classical ensemble within an arbitrary given accuracy. Furthermore, the trajectories of this classical ensemble, generally speaking, need not coincide with the trajectories obtained from de Broglie--Bohmian mechanics. Consequences of this result are discussed. 
\end{abstract}


\maketitle

\section{Introduction}

It is well known that results of quantum mechanics cannot be reproduced by a local and/or context-independent hidden variable theory \cite{Bell1966, Kochen1967}. Celebrated de Broglie--Bohmian mechanics \cite{deBroglie1926, Bohm1952a, Bohm1952b, Bohm1993}, being a non-local and contextual hidden variables theory, gives a consistent way of constructing an ensemble of classical trajectories which reproduces exactly quantum results. This is achieved by adding an extra term called the quantum potential, which depends on the wave function, into the Hamiltonian of the system. 

Is there any other systematic way (i.e., explicitly non de Broglie--Bohmian approach) of constructing an ensemble of classical trajectories that reproduces the quantum probability distribution? The answer to this question is affirmative owing to the recent method developed in Ref. \cite{Coffey2010} that allows one to construct classical trajectories from quantum mechanics using solely the time evolving probability density (not the full wave function) and without assuming or solving any equations of motion. Furthermore, it follows from the construction that these trajectories need not coincide with de Broglie--Bohmian trajectories. The method in Ref. \cite{Coffey2010} employes the geometrical construction of centroidal Voronoi tessellations.

In the current paper, we go further and present an alternative, though complementary, point of view on non de Broglie--Bohmian trajectories. Based on very mild topological assumptions on the configuration space, we show non-constructively (see Theorem \ref{theorem2}) that for any given quantum system evolving unitarily with the Hamiltonian 
\begin{eqnarray}\label{QHamiltonianIntro}
\hat{H} = \hat{\bf p}^2 /(2m) + U({\bf q}), 
\end{eqnarray}
(bold letters denote $D$-dimensional vectors throughout the paper) where the potential $U({\bf q})$ is a sufficiently smooth function, 
one can construct a classical ensemble with the Hamilton function
\begin{eqnarray}\label{ClassicalHailtonianIntro}
\mathrsfs{H} ({\bf p}, {\bf q}) = {\bf p}^2/(2m) + U^{(\infty)} ({\bf q}),
\end{eqnarray}
where the potential $U^{(\infty)}({\bf q})$ coincides with $U({\bf q})$ for almost all ${\bf q}$ (i.e.,  $U^{(\infty)}$ can be different from $U$ only on a measure zero set), such that the square modulus of the wave function in the coordinate (momentum) representation approximately equals the coordinate (momentum) distribution of the classical ensemble within an arbitrary given accuracy. Additionally, the quantum system with the Hamiltonian 
\begin{eqnarray}
\hat{H}^{(\infty)} = \hat{\bf p}^2/ (2m) + U^{(\infty)}({\bf q}),
\end{eqnarray}
is, physically speaking, the same as the original one with the Hamiltonian (\ref{QHamiltonianIntro}). 

Since our proof is non-constructive, the trajectories of the classical ensemble (\ref{ClassicalHailtonianIntro}) are only known to satisfy very basic geometrical (kinematical) assumptions; hence, these trajectories are not related to de Broglie--Bohmian ones in a general case.

In a nutshell, the proof of Theorem \ref{theorem2} is based on the fact that on the one hand, a quantum system does not ``feel'' any measure zero set modifications of the potential due to the dynamical non-locality of quantum mechanics; on the other hand, the corresponding classical system ``feels'' the difference due to the dynamical local nature of classical mechanics. Here, the {\it dynamical non-locality of quantum mechanics} means that during unitary non-relativistic evolution, the value of the wave function in some point at a subsequent time moment depends on the values of the wave function and the potential in all the points at the current time moment. Correspondingly, the {\it dynamical locality of classical mechanics} implies that according to the Hamilton equations, the position of a classical particle at a subsequent time moment depends only on the values of the potential in some small neighbourhood of the current position of the particle (i.e., the subsequent position depends on the force that acts on the particle in the current position). The measure zero set modifications of the potential $U$ are done along curves, which shall be called ``scratches,'' and the procedure of doing such modifications shall be named ``scratching'' the potential $U$. The main technique of the proof is to perform scratching and locate a classical particle such that the scratches become trajectories of the particle. Then, the shapes of the scratches are chosen such that the probability to find the classical particle in some region approximates the analogous probability for the quantum particle.

Strictly speaking, neither the Bell inequalities \cite{Bell1966} nor the Kochen--Specker theorem \cite{Kochen1967} applies to the classical ensemble constructed in Theorem \ref{theorem2} because Theorem \ref{theorem2} does not discuss the measurement process. A complete hidden variable theory can be obtained from Theorem \ref{theorem2} after embedding the measurement process. This should be done in future investigations; now however, we present intuitive arguments why such embedding should not convert Theorem \ref{theorem2} into a local or context-independent hidden variable theory.  As described in details in the proof of Theorem \ref{theorem2}, scratching is performed depending on the actual values of the square modulus of the wave functions at different time moments and at different positions. Hence, the constructed classical ensemble is to violate the Bell inequalities by the same token as de Broglie--Bohmian mechanics does. The contextuality of Theorem \ref{theorem2} is to arise from the fact that the classical ensemble reproduces the square modulus of the wave functions only; it does not provide any information upon other observables.

The rest of the paper is organized as follows: The next section contains exact formulation and rigorous proof of the informally reformulated result. Discussions of physical consequences of Theorem \ref{theorem2}, such as the issue of the discrimination between classical and quantum dynamics based on a posteriori information, a new interpretation of quantum mechanics, and quantization and semi classical approximation, are presented in Sec. \ref{Sec3}. The conclusions are drawn in the last section. Finally, a result upon the Diophantine approximation that is employed in the proof of Theorem \ref{theorem2} is obtained in the Appendix.

\section{Main results}\label{Sec2}

Perhaps, a simplest result that can be obtained within our framework is Theorem \ref{theorem1}. Despite the fact that Theorem \ref{theorem2} is the main result of the paper, we present Theorem \ref{theorem1} due to the following  methodological motivation: The proofs of both the theorems are mainly based on same ideas, however, Theorem \ref{theorem1} allows for more lucid  presentation of these ideas owing to a relative simplicity of its proof.

Let $\Lambda$ denote the Lebesgue measure in $\mathbb{R}^D$ throughout the paper.
\begin{theorem}\label{theorem1}
Assume that $D \geqslant 2$ and
\begin{enumerate}
\item $B \subseteq\mathbb{R}^D$, $0 < \Lambda(B) \leqslant \infty$, is an open convex set;
\item $B_1, \ldots, B_n$ is a closed cover of $B$ such that $B \subset \bigcup_{j=1}^n B_j$, $\Lambda(B) = \sum_{j=1}^n \Lambda(B_j)$, $\Lambda(B_j) \neq 0$, $j=1,\ldots,n$, $\Lambda(B_j \cap B_k) = 0$, $\forall k\neq j$;
\item A quantum system with the Hamiltonian 
\begin{eqnarray}
\hat{H}(t) =\hat{\bf p}^2 /(2m) + U({\bf q}, t),
\end{eqnarray}
is given, where $U({\bf q}, t)$ is a continuously differentiable function $\forall {\bf q}\in B$ and a continuous function $\forall t \in [t_i, t_f]$ ($-\infty < t_i < t_f < +\infty$) and 
\begin{eqnarray}\label{UIntegrable1}
\int_B \left|U({\bf q}, t)\right|  d^D {\bf q} < \infty, \qquad \forall t \in [t_i, t_f];
\end{eqnarray}
\item The wave function of the system, $\Psi ({\bf q}, t)$, evolves according to the Schr\"{o}dinger equation $\forall t \in [t_i, t_f]$, 
\begin{eqnarray}\label{ConditionOfConfinements}
\int_B \left| \Psi({\bf q}, t) \right|^2 d^D {\bf q} =1, \qquad \forall t \in [t_i, t_f].
\end{eqnarray}
Let
\begin{eqnarray}
P_k (t) = \int_{B_k} \left| \Psi({\bf q}, t) \right|^2 d^D {\bf q},
\end{eqnarray}
denote the probability to find a quantum particle in the region $B_k$ at time moment $t$;
\end{enumerate}
Then, for an arbitrary natural number $Q > n^{2n}$, there exists an ensemble of $N$ ($0 < N \leqslant Q$) classical particles with the Hamiltonian function
\begin{eqnarray}\label{ClassicalHamiltonian1}
\mathrsfs{H}^{(\lambda)} ({\bf p}, {\bf q}, t) = {\bf p}^2 /(2m) + U^{(\lambda)} ({\bf q}, t),
\end{eqnarray}
and a quantum system with the wave function $\Psi^{(\lambda)}({\bf q}, t)$ and the Hamiltonian 
\begin{eqnarray}\label{TimeDepSchratchHamiltonian}
\hat{H}^{(\lambda)}(t) = \hat{\bf p}^2 /(2m) + U^{(\lambda)}({\bf q}, t),
\end{eqnarray}
where $U^{(\lambda)}({\bf q}, t)$, $\forall \lambda>0$, is a continuously differentiable function $\forall {\bf q}\in B$ and a continuous function $\forall t \in [t_i, t_f]$, such that the equalities 
\begin{eqnarray}
\lim_{\lambda\to\infty} U^{(\lambda)}({\bf q}, t) &=&U({\bf q}, t), \label{LimPropertySchratchedPotential}\\
\lim_{\lambda\to\infty} \Psi^{(\lambda)} ({\bf q}, t) &=& \Psi ({\bf q}, t), \label{LimPropertySchratchedWaveFunc}
\end{eqnarray}
are valid $\forall t \in [t_i, t_f]$ and for almost all ${\bf q} \in B$ (they may be violated on a measure zero set). Furthermore, there exits the initial condition for the classical ensemble such that
\begin{eqnarray}
\max_{k=1,\ldots,n} \left\{ \left| P_k (t_{i,f}) - \pi_k^{(\lambda)} (t_{i,f}) \right| \right\} < \frac{1}{NQ^{\frac{1}{2n}}}, \, \forall \lambda > 0, \label{ApproximationBoundProbabilities}
\end{eqnarray}
where $\pi_k^{(\lambda)} (t)$ denotes the probability to find a classical particle in the region $B_k$ at time moment $t$, and the equality
\begin{eqnarray}
\sum_{k=1}^n \pi_k^{(\lambda)} (t) = 1, \qquad \forall \lambda > 0 \label{ClassicalConfimentCondition}
\end{eqnarray}
is valid for almost all  $t \in [t_i, t_f]$\footnote{
There may be time moments when the classical particle is on the boundary of one of the regions $B_k$.
},
i.e., the motion of the classical ensemble is confined in the region $B$.
\end{theorem}

\begin{proof}
The main motif of the current paper is the employment of the procedure of {\it ``scratching''} the potential.  We call $U^{(\lambda)}({\bf q}, t)$  the scratched version of the potential $U({\bf q}, t)$ if
\begin{eqnarray}\label{DeffScratchedPotential}
U^{(\lambda)}({\bf q}, t) = U({\bf q}, t) \prod_{k=1}^N \left[ 1 - e^{-\lambda f_k({\bf q})} \right],
\end{eqnarray}
where $f_k({\bf q})$ is continuously differentiable, and the equation $f_k({\bf q})=0$ defines a differentiable curve, which is called the $k^{\rm th}$ scratch, and moreover, we assume that $\partial f_k /\partial q_j = 0$ along the $k^{th}$ scratch. One readily notices that
\begin{eqnarray}\label{AssynptoticRepresPotentialOntheScratch} 
U^{(\lambda)}({\bf q}, t) = \left\{
\begin{array}{ccc}
0 & : & \bigvee_{k=1}^N \left[ f_k({\bf q}) = 0 \right], \\
U({\bf q}, t) + O\left( \lambda^{-\infty} \right) & : & \mbox{otherwise},
\end{array}\right.
\end{eqnarray}
where $\bigvee$ is the logical disjunction and the symbol $O\left( \lambda^{-\infty} \right)$ denotes a term that decays faster than any power of $1/\lambda$ as $\lambda\to\infty$. Equation (\ref{AssynptoticRepresPotentialOntheScratch}) can be rephrased as the potential along the scratches vanishes.

One of most important property of the procedure of scratching is that
\begin{eqnarray}\label{AssynptoticRepresForceOntheScratch}
\frac{\partial U^{(\lambda)}}{\partial q_j} = \left\{
\begin{array}{ccc}
0 & : & \bigvee_{k=1}^N \left[ f_k({\bf q}) = 0 \right], \\
\partial U/\partial q_j + O\left( \lambda^{-\infty} \right) & : & \mbox{otherwise}.
\end{array}\right.
\end{eqnarray}
This equation implies that a classical particle experiences no force along the scratches. Using this idea, we force $N$ classical particles to move along straight lines by choosing
\begin{eqnarray}\label{ScratchForLinearTrajectories}
f_k ({\bf q}) &=& \left( {\bf q} - {\bf q}_k (t_i) - \left\{ {\bf q}_k(t_f) - {\bf q}_k (t_i) \right\} \right.\\
& \times& \left. \left\{ [{\bf q}]_1 - [{\bf q}_k(t_i)]_1 \right\}/\left\{[{\bf q}_k(t_f)]_1 - [{\bf q}_k(t_i)]_1 \right\} \right)^2, \nonumber
\end{eqnarray} 
where $[{\bf a}]_j$ denotes the $j^{\rm th}$ component of a vector ${\bf a}$. Since
\begin{eqnarray}
f_k({\bf q}) = 0 \Longleftrightarrow {\bf q} = {\bf q}_k(t_i) + \left[ {\bf q}_k(t_f) - {\bf q}_k(t_i)\right]\xi, \\
\xi =\left\{ [{\bf q}]_1 - [{\bf q}_k(t_i)]_1 \right\}/\left\{[{\bf q}_k(t_f)]_1 - [{\bf q}_k(t_i)]_1 \right\}, \nonumber
\end{eqnarray}
and also $\xi\left( {\bf q}_k(t_i) \right) = 0$, $\xi\left( {\bf q}_k(t_f) \right) = 1$, 
we indeed conclude that equation $f_k({\bf q}) = 0$ defines the line that goes through the points $ {\bf q}_k(t_i)$ and $ {\bf q}_k(t_f)$; additionally, $\partial f_k /\partial q_j = 0$ on that line.

According to Eq. (\ref{ConditionOfConfinements}), $\sum_{k=1}^n P_k(t_i) = \sum_{k=1}^n P_k(t_f) = 1$; hence, the sequence of real numbers $P_1(t_i), \ldots, P_n(t_i), P_1(t_f), \ldots, P_n(t_f)$ obeys the conditions of the Lemma presented in the Appendix. Therefore, we conclude that $\forall Q \in \mathbb{N}$, $Q > n^{2n}$, there exists (positive) integers $N, N_1(t_i), \ldots, N_n(t_i), N_1(t_f), \ldots, N_n(t_f)$, such that
\begin{eqnarray}\label{ConservationOfNumberClassparticles}
N = \sum_{k=1}^n N_k (t_i) = \sum_{k=1}^n N_k(t_f),
\end{eqnarray}
and Eq. (\ref{ApproximationBoundProbabilities}) is valid if we set 
\begin{eqnarray}\label{DiophantineApproxOfClassicalProb} 
\pi_k^{(\lambda)} (t_{i,f}) = N_k(t_{i,f})/N.
\end{eqnarray}
According to Eqs. (\ref{ConservationOfNumberClassparticles}) and (\ref{DiophantineApproxOfClassicalProb}), it is justifiable to interpret $N$ as the total number of particles in the classical ensemble, and $N_k(t_i)$ [$N_k(t_f)$] as the number of classical particles located in the region $B_k$ at time moment $t_i$ [$t_f$].

Note that there are in fact two possible interpretations for the positive integer $N$. First, we could assume that we have $N$ identical copies (i.e., $N$ realizations) of the one-particle classical system with the Hamiltonian (\ref{ClassicalHamiltonian1}). (This is a usual construction in literature.) Obviously, the particle from one such a copy cannot interact with the particle from another copy. Second, the interpretation utilized here as well as in Theorem \ref{theorem2} is that we assume to have one $N$-particle classical system (i.e., ``gas'' of $N$ particles), where these particles can interact only when they collide with each other. Since we will locate scratches and select the initial conditions to completely eliminate collisions, both these interpretations are indeed equivalent. Nevertheless, the latter one is more general.

Let us arrange the motion of the $N$ particles according to the interpretation. This, however, can be easily achieved by means of the scratched potential [Eq. (\ref{DeffScratchedPotential})] with Eq. (\ref{ScratchForLinearTrajectories}), where ${\bf q}_k(t_{i,f})$ denote the initial and final positions of the $k^{\rm th}$ particle, correspondingly. Note that each particle must be placed in its own scratch. The final problem is to select points ${\bf q}_k(t_{i,f})$ such that no collision between particles is possible. 

Let us construct the finite sets $C_k(t_{i,f})$, $k=1,\ldots,n$, such that
$C_k(t_{i,f}) \subset B_k \setminus\partial B_k$, $\left| C_k(t_{i,f})\right| = N_k(t_{i,f})$,  
$$
\bigcup_{k=1}^n C_k(t_{i,f}) = \left\{ {\bf q}_1(t_{i,f}), \ldots, {\bf q}_N(t_{i,f})\right\},
$$
and the following property is satisfied: 
\begin{eqnarray}\label{PropertyOFCset}
\forall {\bf x}, {\bf y}, {\bf z} \in C  \qquad \forall t \in \mathbb{R} \qquad {\bf z} \neq {\bf x} + t( {\bf y} - {\bf x}),
\end{eqnarray}
where $C = \bigcup_{k=1}^n C_k(t_i) \cup C_k(t_f)$ and $|C| = 2N$ (here, $|A|$ denotes the number of elements of a finite set $A$). In other words, $C$ is a set of all the initial and final positions for the classical particles; $C_k(t_i)$ [$C_k(t_f)$] is a set of all initial [final] positions for particles in the region $B_k$. All $2N$ elements of the set $C$ are distinct points, and no three points from $C$ lie on a line. Such selected points ${\bf q}_k(t_{i,f})$ guarantee not only that no two particles share a scratch, but also that the collisions are prevented (nevertheless only partially). 

We shall prove by contradiction that the described above set $C$ exists.  The justification of the existence of $C$ is based on the fact that $\Lambda(B_k) \neq 0$, $k=0,\ldots,n$.  Let us fix integers $k,l,m=1,\ldots,n$ ($k$, $l$, and $m$ may coincide) and introduce the notation $A^{(1)} = B_k \setminus \partial B_k$,  $B^{(1)} = B_l \setminus \partial B_l$, and $C^{(1)} = B_m \setminus\partial B_m$. The statement 
\begin{eqnarray}\nonumber
\exists {\bf x} \in A^{(1)}, {\bf y} \in B^{(1)} \: \forall {\bf z} \in C^{(1)} \: \exists t \in \mathbb{R} \quad {\bf z} = {\bf x} + t( {\bf y} - {\bf x})
\end{eqnarray}
must be false because otherwise it implies that the set $C^{(1)}$ is a line,  hence $\Lambda\left(C^{(1)}\right) = 0$, which contradicts  assumption 2. Therefore, the following statement is true:
\begin{eqnarray}\nonumber
\forall {\bf x} \in A^{(1)}, {\bf y} \in B^{(1)} \: \exists {\bf z} \in C^{(1)} \: \forall t \in \mathbb{R} \quad {\bf z} \neq {\bf x} + t( {\bf y} - {\bf x})
\end{eqnarray}
Let us pick an arbitrary triple of such points and denote them by ${\bf x}^{(1)}$, ${\bf y}^{(1)}$, and ${\bf z}^{(1)}$, respectively. Now construct the sets $A^{(2)} = A^{(1)} \setminus \left\{ {\bf x}^{(1)} \right\}$, $B^{(2)} =B^{(1)} \setminus \left\{ {\bf y}^{(1)} \right\}$, and $C^{(2)} = C^{(1)} \setminus \left\{ {\bf z}^{(1)} \right\}$. The statement
$$
\exists {\bf x} \in A^{(2)}, {\bf y} \in B^{(2)} \: \forall {\bf z} \in C^{(2)} \: \exists t \in \mathbb{R} \quad {\bf z} = {\bf x} + t( {\bf y} - {\bf x})
$$
is false because  $\Lambda\left( C^{(2)} \right) =  \Lambda\left( C^{(1)} \right)$. One readily observes that this iteration procedure can be continued an arbitrary number of times; hence, the sets $C_k(t_{i,f})$ exist for any $N$.

Recall that ${\bf q}_k(t_{i,f})$ denote the initial and final positions of the $k^{\rm th}$ particle. If the initial momentum of the $k^{\rm th}$ particle is chosen as
$$
{\bf p}_k = m \left[ {\bf q}_k(t_f) - {\bf q}_k(t_i) \right] / (t_f - t_i),
$$ 
then, according to the Hamilton equations, its trajectory is 
$$
{\bf q}_k(t) = {\bf q}_k(t_i) + (t-t_i)\left[ {\bf q}_k(t_f) - {\bf q}_k(t_i)\right] / (t_f - t_i).
$$
Assume that the $j^{\rm th}$ and $l^{\rm th}$ particles collide, viz.,
$\exists t \in [t_i, t_f]$ such that ${\bf q}_j(t) = {\bf q}_l(t)$, i.e., 
\begin{eqnarray}
&& \frac{t-t_i}{t_f - t_i} = \frac{ \left[ {\bf q}_j(t_i) - {\bf q}_l(t_i) \right]_1 }{ \left[ {\bf q}_l(t_f) - {\bf q}_j(t_f) + {\bf q}_j(t_i) - {\bf q}_l(t_i) \right]_1 } \nonumber\\
&&= \ldots = \frac{ \left[ {\bf q}_j(t_i) - {\bf q}_l(t_i) \right]_D }{ \left[ {\bf q}_l(t_f) - {\bf q}_j(t_f) + {\bf q}_j(t_i) - {\bf q}_l(t_i) \right]_D }.
\end{eqnarray}
However, because $B$ is a Hausdorff space, we can always perturb, say $\left[ {\bf q}_j(t_i)\right]_1$, by some small value such that this chain of equalities is violated, but the sets $C_k(t_{i,f})$ preserves their properties.

We note that Eq. (\ref{ClassicalConfimentCondition}) follows from the fact that the trajectories are linear and the convexity of $B$ (assumption 1).

The Schr\"{o}dinger equation in the momentum representation for the scratched quantum system with the Hamiltonian (\ref{TimeDepSchratchHamiltonian}) reads
\begin{eqnarray}\label{SchrodingerEqMomentumRepresentation}
i\hbar \frac{\partial}{\partial t} \Phi^{(\lambda)}({\bf p}, t) = \Phi^{(\lambda)} ({\bf p}, t) {\bf p}^2 / (2m) \quad \quad\quad \nonumber\\
 + \int_{\mathbb{R}^D} d^D {\bf k}\, \widetilde{U^{(\lambda)}} ({\bf k} - {\bf p}, t) \Phi^{(\lambda)}({\bf k}, t),
\end{eqnarray}
where 
\begin{eqnarray}
\Phi^{(\lambda)}({\bf p}, t) &=& \int_B \frac{d^D {\bf q}}{(2\pi\hbar)^{D/2}} \, \Psi^{(\lambda)} ({\bf q}, t) e^{ -i {\bf p}\cdot{\bf q} /\hbar }, \\
\widetilde{U^{(\lambda)}} ({\bf p}, t) &=& \int_B \frac{d^D {\bf q}}{(2\pi\hbar)^{D}} \, U^{(\lambda)} ({\bf q}, t) e^{ i {\bf p}\cdot{\bf q} /\hbar }.
\end{eqnarray}
The integral over ${\bf k}$ in Eq. (\ref{SchrodingerEqMomentumRepresentation}) manifests the dynamical non-locality of quantum mechanics. From Eq. (\ref{AssynptoticRepresPotentialOntheScratch}) and the fact that the value of an integral is unchanged if the integrand is modified on a measure zero set, we conclude that 
\begin{eqnarray}\label{Utilde_lambda}
\widetilde{U^{(\lambda)}} ({\bf p}, t) = \widetilde{U} ({\bf p}, t) + O\left( \lambda^{-\infty}\right).
\end{eqnarray}
Assumption (\ref{UIntegrable1}) guarantees the existence of the Fourier transform of $U$ as well as $U^{(\lambda)}$. Physically speaking, Eq. (\ref{Utilde_lambda}) is a direct consequence of the dynamical non-locality. 

Equations (\ref{LimPropertySchratchedPotential}) and (\ref{LimPropertySchratchedWaveFunc}) follow from the fact that the Fourier transform is ``insensitive'' to measure zero set modifications. 
\end{proof}

\begin{theorem}\label{theorem2}
Assume that $D \geqslant 3$ and
\begin{enumerate}
\item $B\subseteq\mathbb{R}^D$, $0 < \Lambda(B) \leqslant \infty$, is an open path connected set;
\item $B_1, \ldots, B_n$ is a closed cover of $B$ such that $B \subset \bigcup_{j=1}^n B_j$, $\Lambda(B) = \sum_{j=1}^n \Lambda(B_j)$, $\Lambda(B_j) \neq 0$, $j=1,\ldots,n$, $\Lambda(B_j \cap B_k) = 0$, $\forall k\neq j$;
\item A quantum system with the Hamiltonian 
\begin{eqnarray}\label{OriginalQuantumHamiltonian2}
\hat{H} = \hat{\bf p}^2 /(2m) + U({\bf q}),
\end{eqnarray}
is given, where $U({\bf q})$ being a twice continuously differentiable function $\forall {\bf q}\in B$, and
\begin{eqnarray}\label{ConditionOfPossitivityOfU}
U({\bf q}) > 0, \quad \forall {\bf q}\in B, \qquad \int_B \left|U({\bf q})\right|  d^D {\bf q} < \infty;
\end{eqnarray}
\item The wave function of the system, $\Psi ({\bf q}, t)$, evolves according to the Schr\"{o}dinger equation $\forall t \in [t_i, t_f]$ ($-\infty < t_i < t_f < +\infty$) and 
\begin{eqnarray}\label{ConditionOfConfinements2}
\int_B \left| \Psi({\bf q}, t) \right|^2 d^D {\bf q} =1, \qquad \forall t \in [t_i, t_f].
\end{eqnarray}
Let
\begin{eqnarray}\label{Deff_Pk_2}
P_k (t) = \int_{B_k} \left| \Psi({\bf q}, t) \right|^2 d^D {\bf q},
\end{eqnarray}
denote the probability to find a quantum particle in the region $B_k$ at time moment $t$;
\item $\widetilde{B}_1, \ldots, \widetilde{B}_n$ is a closed cover of $\mathbb{R}^D$ such that $\mathbb{R}^D = \bigcup_{j=1}^n \widetilde{B}_n$, $\Lambda\left(\widetilde{B}_j \right) \neq 0$, $j=1,\ldots,n$, $\Lambda\left( \widetilde{B}_j \cap \widetilde{B}_k \right) = 0$, $\forall k\neq j$, and
\begin{eqnarray}\label{Deff_tilde_Pk_2}
\widetilde{P}_k (t) = \int_{\widetilde{B}_k} \left| \Phi({\bf p}, t) \right|^2 d^D {\bf p},
\end{eqnarray}
being the probability that the momentum of the quantum particle is in the region $\widetilde{B}_k$ and $\Phi ({\bf p}, t)$ is the wave function in the momentum representation;
\item $t_i = t_1 < t_2 < \ldots < t_{K-1} < t_{K} = t_f$;
\end{enumerate}
Then, for an arbitrary natural number $Q > n^{2Kn}$, there exists an ensemble of $N$ ($0 < N \leqslant Q$) classical particles with the Hamiltonian function
\begin{eqnarray}\label{ClassicalSystem2}
\mathrsfs{H}^{(\lambda)} ({\bf p}, {\bf q}) = {\bf p}^2 /(2m) + U^{(\lambda)} ({\bf q}),
\end{eqnarray}
and a quantum system with the wave function $\Psi^{(\lambda)}({\bf q}, t)$ and the Hamiltonian 
\begin{eqnarray}
\hat{H}^{(\lambda)} = \hat{\bf p}^2 /(2m) + U^{(\lambda)}({\bf q}),
\end{eqnarray}
where $U^{(\lambda)}({\bf q})$, $\forall \lambda>0$, is a twice continuously differentiable function $\forall {\bf q}\in B$, such that the equalities 
\begin{eqnarray}
\lim_{\lambda\to\infty} U^{(\lambda)}({\bf q}) &=&U({\bf q}), \label{LimPropertySchratchedPotential2}\\
\lim_{\lambda\to\infty} \Psi^{(\lambda)} ({\bf q}, t) &=& \Psi ({\bf q}, t), \label{LimPropertySchratchedWaveFunc2}
\end{eqnarray}
are valid $\forall t \in [t_i, t_f]$ and for almost all ${\bf q} \in B$ (Eqs. (\ref{LimPropertySchratchedPotential2}) and (\ref{LimPropertySchratchedWaveFunc2}) may be violated on a measure zero set). Furthermore, there exits the initial condition for the classical ensemble such that
\begin{eqnarray}
\max_{k=1,\ldots,n} \left\{ \left| P_k (t_j ) - \pi_k^{(\infty)} (t_j ) \right| \right\} < \frac{1}{NQ^{\frac{1}{2Kn}}}, \label{ApproximationBoundProbabilities2}\\
\max_{k=1,\ldots,n} \left\{ \left| \widetilde{P}_k (t_j ) - \widetilde{\pi}_k^{(\infty)} (t_j ) \right| \right\} < \frac{1}{NQ^{\frac{1}{2Kn}}},  \label{ApproximationBoundProbabilitiesMomenta2}\\
 \forall j=1,\ldots,K,  \nonumber
\end{eqnarray}
where $\pi_k^{(\lambda)} (t)$ denotes the probability to find a classical particle in the region $B_k$ at time moment $t$ $\left[ \pi_k^{(\infty)}(t) = \lim_{\lambda\to\infty} \pi_k^{(\lambda)} (t) \right]$ and $\widetilde{\pi}_k^{(\lambda)} (t)$ denotes the probability that the momentum of the classical particle is in the region $\widetilde{B}_k$ at time moment $t$ $\left[ \widetilde{\pi}_k^{(\infty)}(t) = \lim_{\lambda\to\infty}  \widetilde{\pi}_k^{(\lambda)} (t) \right]$. Additionally, the equalities 
\begin{eqnarray}
\sum_{k=1}^n \pi_k^{(\infty)} (t)=1, \qquad \sum_{k=1}^n \widetilde{\pi}_k^{(\infty)} (t) =1,  \label{ClassicalConfimentCondition2}
\end{eqnarray}
are valid for almost all  $t \in [t_i, t_f]$\footnote{
There may be time moments when the classical particle (the momentum of the classical particle) is on the boundary of one of the region $B_k$ ($\widetilde{B}_k$).
}, i.e., the motion of the classical ensemble is confined in the regions $B$.
\end{theorem}
\begin{proof}
The main idea of the proof is, again, the scratching procedure. Assume that all the scratches given implicitly in the region $B$, i.e., the set of solutions of the following system of equations
\begin{eqnarray}\label{ImplicitDeffScratches2}
F_1^{(l)} ({\bf q}) = 0, \ldots, F_{D-1}^{(l)} ({\bf q}) = 0, 
\end{eqnarray}
where each function is continuously differentiable, defines the $l^{\rm th}$ scratch ($l=1,\ldots,N$). We assume further that there exists integer $j$ ($1\leqslant j \leqslant D$) such that the Jacobi matrix
$$
J_l(j) = \frac{\partial\left( F_1^{(l)}, \ldots, F_{D-1}^{(l)} \right)}{\partial\left( q_1, \ldots, q_{j-1}, q_{j+1}, \ldots, q_D \right)}
$$ 
is non singular on the $l^{\rm th}$ scratch. We point out that it is a natural assumption because according to the implicit function theorem, if $J_l(j)$ is non singular then Eq. (\ref{ImplicitDeffScratches2}) defines (locally) a curve.

Defying auxiliary functions
\begin{eqnarray}\label{DeffAuxiliaryf_2}
f_l ({\bf q}) = \sum_{j=1}^{D-1} \left[ F_j^{(l)}({\bf q}) \right]^2,
\end{eqnarray}
we introduce the scratched potential 
\begin{eqnarray}\label{ScratchedPotential2}
U^{(\lambda)} ({\bf q}) = U({\bf q}) \prod_{l=1}^N \left[ 1 - e^{-\lambda f_l({\bf q})} \right].
\end{eqnarray}
One readily verifies from definition (\ref{DeffAuxiliaryf_2}) that $f_l ({\bf q} ) =0$ and $\partial f_l({\bf q}) /\partial {\bf q} = {\bf 0}$ if and only if ${\bf q}$ lines on the $l^{\rm th}$ scratch. Setting $F_D^{(l)}({\bf q}) \equiv 0$, $l = 1,\ldots,N$, let us scrutinize the Hessian matrix of the function $f_l ({\bf q} )$ on the $l^{\rm th}$ scratch,
\begin{eqnarray}
\left. H(f_l) \right|_l &=& \left( \left. \frac{\partial^2 f_l}{\partial q_r\partial q_k} \right|_l \right) = 2\left( \sum_{j=1}^D \frac{\partial F_j^{(l)}}{\partial q_r} \frac{\partial F_j^{(l)}}{\partial q_k} \right) \nonumber\\
&& = 2\left( {\bf x}_k^{(l)} \cdot {\bf x}_r^{(l)} \right) \label{GrammianRepresentationF}\\
&& = \left. 2J_l^T J_l \right|_l, \label{ProductOfjacobianRepresentationF}
\end{eqnarray}
where ${\bf x}_k^{(l)} = \left( \partial F_1^{(l)}/\partial q_k, \ldots, \partial F_D^{(l)}/\partial q_k  \right)$ and 
$$
J_l = \partial\left( F_1^{(l)}, \ldots, F_D^{(l)} \right) / \partial\left( q_1, \ldots, q_D \right).
$$
On the one hand, according to Eq. (\ref{GrammianRepresentationF}), $\left. H(f_l) \right|_l$ is a Gramian matrix; thus,  $\left. H(f_l) \right|_l$ is positive semidefinite. On the other hand, Eq. (\ref{ProductOfjacobianRepresentationF}) manifests that the Hessian $\left. H(f_l) \right|_l$ is of $D-1$ rank because $J_l$ is singular (since $F_D^{(l)}=0$) and the following $D-1$ minor of the Hessian:
$$
\left(  \left. \frac{\partial^2 f_l}{\partial q_r\partial q_k} \right|_l, \, {r \neq j \atop k \neq j} \right) = \left. 2J_l^T(j) J_l(j) \right|_l
$$  
is not singular (since $J_l(j)$ is non singular). Hence, we conclude that the Hessian $\left. H(f_l) \right|_l$ has one zero eigenvalue and $D-1$ positive eigenvalues. 

We now show that the eigenvector corresponding to the zero eigenvalue is the tangent vector of the scratch. Let ${\bf q} = {\bf q}^{(l)}(s)$ be a parametric representation of the $l^{\rm th}$ scratch. Then, 
$$
F_1^{(l)} \left( {\bf q}^{(l)}(s) \right) = 0, \ldots, F_{D-1}^{(l)} \left( {\bf q}^{(l)}(s) \right) = 0,
$$
and consequently,
$
d F_1^{(l)} / ds = 0, \ldots, d F_{D-1}^{(l)} /ds = 0
$.
Since 
\begin{eqnarray}
\frac{ d{\bf q}^{(l)} }{ds} \cdot \left. H(f_l) \right|_l \left( \frac{ d{\bf q}^{(l)} }{ds} \right)^T = 2\sum_{j=1}^{D-1} \left( \frac{d F_j^{(l)}}{ds} \right)^2 = 0,
\end{eqnarray}
and the fact that the vector $d{\bf q}^{(l)} / ds$ is proportional to the tangent vector of the $l^{\rm th}$ scratch, we have confirmed our statement. 

Equations (\ref{AssynptoticRepresPotentialOntheScratch}) and (\ref{AssynptoticRepresForceOntheScratch}) are also valid in the case of the potential (\ref{ScratchedPotential2}). In other words, the scratches are equilibrium positions. Therefore, to analyze dynamics of the classical system (\ref{ClassicalSystem2}), we can linearize the Hamilton equations in a neighbourhood of the $l^{\rm th}$ scratch (such a procedure is  justified by the theorem on page 100 of Ref. \cite{Arnold1989}).

The Hessian of the scratched potential (\ref{ScratchedPotential2}) on the $l^{\rm th}$ scratch reads 
\begin{eqnarray}\label{HessianScratchedPotential2}
\left. H\left( U^{(\lambda)} \right) \right|_l &=& \lambda U \left. H(f_l) \right|_l \prod_{j=1, \: j \neq l}^N \left( 1-e^{-\lambda f_j} \right).
\end{eqnarray}
To understand the qualitative behaviour of the classical system, one needs to analyze the eigenvalues of the Hessian (\ref{HessianScratchedPotential2}) (see the theorem on page 104 of Ref. \cite{Arnold1989}). Recalling the properties of $\left. H(f_l) \right|_l$ and inequality (\ref{ConditionOfPossitivityOfU}), we conclude that the Hessian (\ref{HessianScratchedPotential2}) has $D-1$ positive eigenvalues and one zero eigenvalue, with the corresponding eigenvector being the tangent vector of the $l^{\rm th}$ scratch, if no scratches intersect and $\lambda >0$.

According to the theorem on page 76 of Ref. \cite{Arnold1989}, if  a classical particle is placed in a scratch, then in the limit $\lambda\to\infty$, the motion of the particle is constrained to the scratch and its velocity can be only collinear to the tangent vector of the scratch. We are now in the position to formulate {\it the fundamental property of the scratches:} the procedure of scratching a potential of some unconstrained classical mechanical system (implicitly) introduces holonomic constrains in the limit $\lambda\to\infty$. However, these constrains can be ``felt'' by a classical particle only if it is placed in one of the scratches, thus, the name {\it implicit holonomic constrains}. 

The same conclusion can be reached intuitively as follows: If a classical particle, being initially in a scratch, does not move along the scratch, then according to Eq. (\ref{HessianScratchedPotential2}), the particle experiences the force of the magnitude proportional to $\lambda$ that pushes the particle back to the scratch. Hence, the motion of the particle evidently becomes constrained in the limit $\lambda\to\infty$.

The sequence of real numbers $P_1(t_1), \ldots, P_n(t_1), \ldots, P_1(t_K), \ldots, P_n(t_K)$, $\widetilde{P}_1(t_1), \ldots, \widetilde{P}_n(t_1), \ldots, \widetilde{P}_1(t_K), \ldots, \widetilde{P}_n(t_K)$ obeys the conditions of the Lemma from the Appendix because 
$\sum_{k=1}^n P_k(t_j) = \sum_{k=1}^n \widetilde{P}_k(t_j) = 1$, $j=1,\ldots,K$. Therefore, $\forall Q \in \mathbb{N}$, $Q>n^{2Kn}$,  there exists (positive) integers $N, N_1(t_1), \ldots, N_n(t_1), \ldots, N_1(t_K), \ldots, N_n(t_K)$, $\widetilde{N}_1(t_1), \ldots, \widetilde{N}_n(t_1), \ldots, \widetilde{N}_1(t_K), \ldots, \widetilde{N}_n(t_K)$ such that inequalities (\ref{ApproximationBoundProbabilities2}) and (\ref{ApproximationBoundProbabilitiesMomenta2}) hold if
\begin{eqnarray}\label{ClassicalProbabilityRepresentation2}
\pi_k^{(\infty)}(t_j) = N_k (t_j) / N, \qquad \widetilde{\pi}_k^{(\infty)}(t_j) = \widetilde{N}_k (t_j) / N, \nonumber\\
j=1,\ldots,K, \qquad k=1,\ldots,n.
\end{eqnarray}
Similarly to Theorem \ref{theorem1}, we interpret $N$ as the total number of classical particles in the ensemble, $N_k(t_j)$ -- the number of classical particles in the region $B_k$ at time moment $t_j$, and $\widetilde{N}_k(t_j)$ -- the number of classical particles with momenta in the region $\widetilde{B}_k$ at time moment $t_j$.

For the given natural numbers $N_k(t_j)$, we can construct simple (i.e., non self-intersecting) paths ${\bf q}^{(1)}(s), \ldots, {\bf q}^{(N)}(s) : [0,1] \to B$, such that there exist real numbers, $0=s_1 < s_2 < \ldots < s_K=1$, and
\begin{eqnarray}\label{ConditionOnIntersectionBk2}
\left| \left\{ {\bf q}^{(1)}(s_j), \ldots, {\bf q}^{(N)}(s_j)\right\} \bigcap \left( B_k \setminus\partial B_k \right)\right| = N_k(t_j),
\end{eqnarray}
where $k=1,\ldots,n$ and $j=1,\ldots,K$. Moreover, we put an additional requirement that no such two paths can intersect. We shall demonstrate that such a construction exists.

Due to the path connectedness of $B$, we can pick any finite set of distinct points in $B$ and connect them by a path. Now we demonstrate that since $B$ is an open set, a path connecting a finite set of distinct points can be chosen to be a simple one. Let us select any path that goes thought these points and denote it by ${\bf q} : [0,1] \to B$. Assume that the path self-intersects in the point ${\bf x}$. Since $B$ is open and Hausdorff, we can select an open ball centred at the point of self-intersection, $O_{\varepsilon}({\bf x}) = \{ {\bf y} \in \mathbb{R}^D \, | \, \| {\bf y} - {\bf x} \| < \varepsilon\}$, such that $O_{\varepsilon}({\bf x}) \subset B$ and the radius of the ball is chosen such that the following condition satisfies
$$
\left\{ {\bf q}(s) \, | \, 0 \leqslant s \leqslant 1 \right\} \bigcap \partial O_{\varepsilon}({\bf x}) = \left\{ {\bf q}(s_1), {\bf q}(s_2), {\bf q}(s_3), {\bf q}(s_4) \right\},
$$ 
where $s_1 < s_2 < s_3 < s_4$ and ${\bf q}(s_1) \neq {\bf q}(s_2) \neq {\bf q}(s_3) \neq {\bf q}(s_4)$. In other words, the open ball contains two segments $G_1 = \{ {\bf q}(s) \, | \, s_1 < s < s_2 \} \subset O_{\varepsilon}({\bf x})$ and $G_2 = \{ {\bf q}(s) \, | \, s_3 < s < s_4 \} \subset O_{\varepsilon}({\bf x})$ which intersect $G_1 \cap G_2 = \{ {\bf x} \}$. We shall use the convexity of $O_{\varepsilon}({\bf x})$ and replace $G_{1,2}$ by linear segments that do not intersect. Let $L({\bf a}, {\bf b}) = \{ {\bf a} + t({\bf b} - {\bf a}) \, | \, 0\leqslant t \leqslant 1\}$ denote a linear segment connecting points ${\bf a}$ and ${\bf b}$. The following statement is true
\begin{eqnarray}
\forall {\bf b}_1, {\bf b}_2, {\bf e}_1, {\bf e}_2 \in \partial O_{\varepsilon}({\bf x}) \: ({\bf b}_1 \neq {\bf b}_2 \neq {\bf e}_1 \neq {\bf e}_2) \: \exists {\bf c} \in O_{\varepsilon}({\bf x})  \nonumber\\
\quad \left[ L({\bf b}_1, {\bf x}) \cup L({\bf x}, {\bf e}_1) \right]\bigcap 
\left[ L({\bf b}_2, {\bf c}) \cup L ({\bf c}, {\bf e}_2) \right] = \emptyset, \nonumber
\end{eqnarray}
because its negation contradicts the fact that $O_{\varepsilon}({\bf x})$ is an open ball (note that the assumption $D\geqslant 3$ is crucial). Substituting $G_{1,2}$ by $L[{\bf q}(s_1), {\bf x}] \cup L[{\bf x}, {\bf q}(s_2)]$ and $L[{\bf q}(s_3), {\bf c}]\cup L[{\bf c}, {\bf q}(c_4)]$, correspondingly, we have demonstrated that there exists a simple path that connects a finite number of distinct points.

Now we shall demonstrate that $B' = B\setminus p({\bf a}, {\bf b})$ is an open path connected set, where $p({\bf a}, {\bf b}) = \{ {\bf q}(s) \, | \, 0 \leqslant s \leqslant 1\} \subset B$ denotes a simple path connecting ${\bf a}, {\bf b} \in B$. Construct the set 
$W_{\varepsilon}  = \{ {\bf y} \in \mathbb{R}^D \, | \, \| {\bf y} - {\bf q}(s) \| < \varepsilon, \, 0 \leqslant s \leqslant 1 \}$. Since $B$ is open, we fix $\varepsilon >0$ such that $\partial W_{\varepsilon} \subset B$. Let us choose arbitrary points ${\bf c}, {\bf d} \in B'$. There are two possibilities -- either $p({\bf c}, {\bf d}) \subset B'$ or $p({\bf c}, {\bf d}) \not\subset  B'$. The latter case means that $p({\bf c}, {\bf d}) \subset B$ and $p({\bf c}, {\bf d}) \cap \partial W_{\varepsilon'} \neq \emptyset$ for some $\varepsilon'$ ($0<\varepsilon' \leqslant \varepsilon$). However, all points in the set $p({\bf c}, {\bf d}) \cap \partial W_{\varepsilon'}$ can be path connected with all the paths lying in $\partial W_{\varepsilon'}$ because $\partial W_{\varepsilon'}$ is a path connected set. Finally, since $ \partial W_{\varepsilon'} \subset B'$, we have demonstrated that $B'$ is an open path connected set. By the same token, $\forall {\bf f}, {\bf g} \in B'$, $B'' = B' \setminus p({\bf f}, {\bf g})$ (where $p({\bf f}, {\bf g})\subset B'$) is an open path connected set as well. In other words, we can subtract any finite number of simple non-intersecting paths from $B$ without affecting its path connectedness. Therefore, ${\bf q}^{(1)}(s), \ldots, {\bf q}^{(N)}(s)$ exists with the prescribed properties.

The paths ${\bf q}^{(l)}(s)$ are continuous curves by definition. Thus, due to the Weierstrass theorem, they can be approximated polynomially within an arbitrary given accuracy [and preserving property (\ref{ConditionOnIntersectionBk2})]. Let ${\bf q}^{(1)}(s), \ldots, {\bf q}^{(N)}(s)$ denote these polynomial curves. Having alter the $l^{\rm th}$ curve in an arbitrary small neighbourhood of the point ${\bf q}^{(l)}(s_j)$, $\forall j=1,\ldots,K$ and $\forall l=1,\ldots,N$, such that  
\begin{eqnarray}
&& \left|  \left\{ \frac{ d{\bf q}^{(l)}(s_j)/ds }{ \left\| d{\bf q}^{(l)}(s_j)/ds \right\| }  \right\}_{l=1,\ldots,N} \label{GettingRightTangetVector2}
 \bigcap \Delta_k \right| = \widetilde{N}_k(t_j), \\
&& \Delta_k = \left\{ {\bf x}/\| {\bf x} \| \: \left| \: {\bf x}\neq {\bf 0} \mbox{ and }  {\bf x} \in \widetilde{B}_k \setminus\partial \widetilde{B}_k \right\} \right. 
\end{eqnarray}
(viz., the direction of the tangent vector of ${\bf q}^{(l)}(s)$ at $s=s_j$ coincides with the direction of one of the momentum vectors from the set $\widetilde{B}_k$), we can use these new curves as the scratches, ${\bf q}^{(1)}(s), \ldots, {\bf q}^{(N)}(s)$, to enforce desired dynamics of the classical ensemble. 

Now let us slightly generalize the procedure of scratching the potential $U({\bf q})$ by performing the following substitution 
\begin{eqnarray}\label{ScratchedModifiedPotential}
U^{(\lambda)} ({\bf q}) = U({\bf q}) \prod_{l=1}^N \left[ 1 - e^{-\lambda f_l({\bf q})} \right] + \sum_{l=1}^N V_l e^{-\lambda f_l ({\bf q})}.
\end{eqnarray}
The fundamental property of the scratches, derived originally for the case of Eq. (\ref{ScratchedPotential2}), is valid in the case of the modified scratching [Eq. (\ref{ScratchedModifiedPotential})] if the potential $V_l$ is a function of the curvilinear coordinate that is in the direction of the $l^{\rm th}$ scratch (see the theorem on page 76 of Ref. \cite{Arnold1989}). Here, the requirement that the scratches must be simple curves that do not intersect leads to an important conclusion that all the functions $V_l$ can be defined $\forall {\bf q} \in B$ (evidently that the function $V_l$ cannot be defined at points where two scratches intersect and at points were a scratch self-intersects).

Now we employ the formalism of the Lagrange equations of the second kind (see, e.g., Ref. \cite{Papastavridis2002}). Assuming that the parametric definition of the $l^{\rm th}$ scratch reads ${\bf q} = {\bf q}^{(l)}\left(s^{(l)} \right)$, it is natural to use $s^{(l)}$ as the generalized coordinate by interpreting this parametric form of the scratch as a holonomic constrain. The Lagrange equations of the second kind for the generalized coordinate $s^{(l)}$ is as follows
\begin{eqnarray}
m \frac{d}{dt}\left[ \left( \frac{d{\bf q}^{(l)}}{ds^{(l)}} \right)^2 \frac{ds^{(l)}}{dt} \right] = -\frac{dV_l}{ds^{(l)}}, \quad l=1,\ldots,N.
\end{eqnarray}

The final step of the proof is to construct the potentials $V_l$ such that the interpretations of $\pi_k^{(\infty)}(t_j)$ and $\widetilde{\pi}_k^{(\infty)}(t_j)$ given by Eq. (\ref{ClassicalProbabilityRepresentation2}) are correct. If we require that 
\begin{eqnarray}
s^{(l)}(t_j) = s_j, \label{ConditionSAtTj2}
\end{eqnarray}
then according to Eq. (\ref{ConditionOnIntersectionBk2}), dynamics of the classical ensemble indeed realizes the interpretation of $\pi_k^{(\infty)}(t_j)$.
Recall that the momentum of the $l^{\rm th}$ classical particle is 
$$
{\bf p}^{(l)}(t_j) = m \frac{ d{\bf q}^{(l)} }{dt} (t_j)= m \frac{ d{\bf q}^{(l)} }{ds^{(l)}} \frac{ds^{(l)}}{dt} (t_j).
$$
Since we have already chosen suitable directions of the momenta in Eq. (\ref{GettingRightTangetVector2}), we are only left to find proper magnitudes of the momenta. $\forall {\bf x} \in\Delta_k$, $\exists c \in\mathbb{R}$ $c{\bf x} \in \widetilde{B}_k \setminus\partial \widetilde{B}_k$; whence, there exit $c_{l,j}  \in\mathbb{R}$ such that 
\begin{eqnarray}
\left| \left\{ m c_{l,j} \frac{d{\bf q}^{(l)}(s_j)}{ds^{(l)}} \right\}_{l=1,\ldots,N} \bigcap \left( \widetilde{B}_k \setminus\partial \widetilde{B}_k \right) \right| = \widetilde{N}_k(t_j), \label{RightMomentumMagnitudeCondition2}\\
 j=1,\ldots,K, \qquad k=1,\ldots,n. \nonumber
\end{eqnarray}
If $d s^{(l)}(t_j)/dt = c_{l,j}$, then the interpretation of $\widetilde{\pi}_k^{(\infty)}(t_j)$ is realized.

Therefore, in order to finalize the proof, we need to find the potentials $V_l$ by knowing that $s^{(l)}(t_j) = s_j$ and $d s^{(l)}(t_j)/dt = c_{l,j}$. These type of problems are known as inverse problems for  ordinary differential equations, and as demonstrated in Ref. \cite{Kunze1999} (see theorem 4 of Ref. \cite{Kunze1999}), the potentials $V_l$ can be found such that $s^{(l)}(t_j) \approx s_j$ and $d s^{(l)}(t_j)/dt \approx c_{l,j}$ within an arbitrary given accuracy.

Conditions (\ref{ClassicalConfimentCondition2}) follow from the construction of the trajectories of the classical particles. Equations (\ref{LimPropertySchratchedPotential2}) and (\ref{LimPropertySchratchedWaveFunc2}) are proven in the same manner as Eqs. (\ref{LimPropertySchratchedPotential}) and (\ref{LimPropertySchratchedWaveFunc}) from Theorem \ref{theorem1}.
\end{proof}

\section{Physical Interpretations of the Result}\label{Sec3}

Main physical consequences of Theorem \ref{theorem2} are the following:

\subsection{Is the system in the black box  quantum or classical?}\label{Sec3s1}

Imagine that some experimentalist is given a black box, and he or she needs to determined whether a quantum system or an ensemble of classical particles is inside the box. The Hamiltonian (i.e., the Hamiltonian operator in the quantum case and the corresponding Hamiltonian function in the classical  case) of the system inside is assumed to be known a priori. The experimentalist is only allowed to measure the coordinate and/or momentum distributions at arbitrary time moments. Then, Theorem \ref{theorem2} negatively answers this question, i.e., having only coordinate and/or momentum distributions for some Hamiltonian system, one cannot conclude whether they are obtained as the squared modulus of a wave function or as distribution functions of an ensemble of classical particles. 

Let us explain how such a conclusion is reached from Theorem \ref{theorem2}. The set $B$ in the theorem in fact represents the given black box; thus, we assume that the topology of the black box obeys assumption 1 of Theorem \ref{theorem2} (indeed, this assumption is valid for a majority of realistic experimental setups). It is crucial that no-matter how precise experimental equipment is, there is aways a finite accuracy in each and every experimental measurement. Additionally, no probability distribution of a continuous variable (such as coordinates and momenta) can be  measured, strictly speaking. What one measures in such a situation is the probability that a value of the variable lies in some small (but finite) region. In terms of Theorem \ref{theorem2}, this means that neither $\left| \Psi({\bf q}, t) \right|^2$ nor $\left| \Phi({\bf p}, t) \right|^2$ is measurable; however, both $P_k (t)$ [Eq. (\ref{Deff_Pk_2})] and $\widetilde{P}_k (t)$ [Eq. (\ref{Deff_tilde_Pk_2})] are measurable in a real-life experiment. At this moment, it is convenient to informally reformulate Theorem \ref{theorem2} from the point of view of the experimentalist: Say that he or she first assumes that the system in the black box is quantum, then the collected data ought to be interpreted as $P_k (t)$ and  $\widetilde{P}_k (t)$, correspondingly. (Note that any real device is capable of taking only a finite number of ``shots,'' so $t$ is a discrete parameter.) However, the theorem says that the original potential energy can be slightly modified such that no real (i.e., of a finite precision) experiment can detect a difference [Eqs. (\ref{LimPropertySchratchedPotential2}) and (\ref{LimPropertySchratchedWaveFunc2})] between the original and altered potential. Moreover, one can construct an ensemble of classical particles (with the Hamiltonian function of the natural form where the altered potential energy used as a potential term [Eq. (\ref{ClassicalSystem2})]) which reproduces the measured data $P_k (t)$ and  $\widetilde{P}_k (t)$ within an arbitrary desired precision [Eqs. (\ref{ApproximationBoundProbabilities2}) and (\ref{ApproximationBoundProbabilitiesMomenta2})]. Thus, the negative answer.

The question whether a given dynamical process being quantum or classical is of active interest in the field of quantum control (see, e.g., Ref. \cite{Franco2010}). Reiterating our results, we conclude that almost any given quantum evolution can be reproduced by means of a classical ensemble with a scratched potential; even the interference in the double-slit experiment can be obtained in this way.

\subsection{Theorem \ref{theorem2} as a new interpretation of quantum mechanics}

De Broglie-Bohmian mechanics gives a consistent way of constructing a classical ensemble which reproduces exactly quantum results. This is achieved by adding an extra term, which depends solely on the wave function, into the potential energy of the system. Theorem \ref{theorem2} gives an alternative way of construction a classical ensemble that reproduces quantum results, however, not exactly but within an arbitrary given accuracy. Contrary to the scratching procedure introduced in the current paper, the de Broglie--Bohmian  modifications of the potential are quite ``noticeable,'' i.e., they are done not on a measure zero set; furthermore, the obtained classical ensemble is unconstrained. As it was discussed in the proof of Theorem \ref{theorem2} (see the fundamental property of scratches), the scratches are in fact holonomic constrains. However, they are not usual constrains (we used the name ``implicit holonomic constrains'')  since a classical particle must have  peculiar initial conditions in order to be able to experience their presence. 

In summary, an ensemble of classical particles can simulate any quantum system if the potential energy of the system is properly adjusted (according to Bohmian mechanics), or if (implicit holonomic) constrains are assigned to the system (according to Theorem \ref{theorem2}). Even though the scratching procedure is dual to de Broglie--Bohmian mechanics, they both share a common feature -- they are nonlocal hidden variable theories. The hidden variables  in our approach are the scratches that cannot be experimentally detected due to their ``measure zero set'' nature. 

\subsection{Quantization and quasi classical approximation}

Before presenting the connection between Theorem \ref{theorem2} and the procedures of quantization and quasi classical approximation, we recall the definition of the Hilbert space $\mathrsfs{L}^2 (B)$. The axiom of quantum mechanics postulates that elements in $\mathrsfs{L}^2 (B)$ represent quantum states. First, we construct the vector space 
$
L^2 (B) = \left\{ f : B \to \mathbb{C} \: \left| \: f \mbox{ is measurable and $|f|^2$  integrable} \right\} \right. ,
$
where measuring and integration is done with respect to the Lebesgue measure. The sesquilinear form 
$
\langle f, g \rangle = \int_B f({\bf q}) \overline{g({\bf q})} \Lambda( d{\bf q} )
$ 
is not an inner product on $L^2 (B)$  because the equality $\langle f, f \rangle = 0$ implies that $f$ vanishes almost everywhere (a. e.), but does not necessary mean that $f\equiv 0$. Having introduced the space,
$
\mathrsfs{N} = \left\{ f\in L^2 (B) \: | \: f = 0 \mbox{ a. e. on $B$} \right\}
$,
we define $\mathrsfs{L}^2 (B) =  L^2 (B) / \mathrsfs{N}$ as a factor space. One can demonstrate that 
$
\left( [f], [g] \right) = \int_B f({\bf q}) \overline{g({\bf q})} \Lambda( d{\bf q} ),
$
$f \in [f]$ and $g\in [g]$, is indeed an inner product in $\mathrsfs{L}^2 (B)$. Here, $[f]$ and $[g]$ are equivalence classes of functions in $L^2 (B)$.

In the light of the presented above definition, Eq. (\ref{LimPropertySchratchedWaveFunc2}) implies that $\Psi^{(\infty)}$ and $\Psi$ correspond to the same physical state. By the same token, Eq. (\ref{LimPropertySchratchedPotential2}) means that, as far as quantum mechanics is concerned, the potentials $U^{(\infty)}$ and $U$ correspond to the same quantum system; however, they correspond to two qualitatively different classical systems.

Hence, generally speaking, quantization maps a classical system to an equivalence class of quantum systems modulo scratching. Since the procedure of quantization and quasi classical approximation are closely related \cite{Shirokov1976b, Shirokov1979d}, the ``dual'' statement is also true: quasi classical approximation maps a quantum system to an equivalence class of a classical systems modulo scratching. 

We reiterate that scratching may tremendously change the behaviour of a classical system. For example, consider a system that contains a potential barrier such that a classical particle with some fixed energy cannot go over the barrier. The quantum counterpart of the classical particle of course can ``go through'' the barrier. To emphasize the conceptual difference in the behaviours of the classical and quantum systems, we use the term ``quantum tunnelling.'' Nonetheless, one can readily scratch the barrier such that the classical particle can also ``go through'' the barrier, but at the same time, this alteration of the potential has no effect on the behaviour of the quantum counterpart.

Physically speaking, the fact that scratching does not affect the behaviour of a quantum system is due to the dynamical non-local character of quantum mechanics. Correspondingly, the dynamical local character of classical mechanics is responsible for the effect that scratching has on a classical system.

\section{Concluding remarks} 

In Theorem \ref{theorem2}, we have demonstrated that a classical ensemble can be constructed that reproduces within an arbitrary accuracy the coordinate and momentum probability distributions of a given quantum system, and the potential energy terms in the Hamiltonians of the classical and quantum systems almost coincide. The consequences of the existence of such a construction have been discussed. It is important that the trajectories of the classical ensemble may differ from the trajectories obtained from de Broglie--Bohmian mechanics. 

Theorem \ref{theorem2} complements results of recent paper \cite{Coffey2010}. In fact, Theorem \ref{theorem2} can be proven constructively employing the geometrical construction of centroidal Voronoi tessellations in the same manner as it was utilized in Ref.  \cite{Coffey2010}. In such a case, the regions $B_j$ from assumption 2 of Theorem \ref{theorem2} should be substituted by the Voronoi cells, $C_j$, see Eq. (2) of Ref. \cite{Coffey2010}; the points ${\bf q}^{(l)}(t_j)$ should be substituted by the centroids of the corresponding Voronoi cells, ${\bf x}_l = {\bf x}_l(t_j)$, see Eq. (3) of Ref. \cite{Coffey2010}. Then, the trajectories constructed by means of the centroidal Voronoi tessellations can be used as scratches. However, note that the current proof of Theorem \ref{theorem2} is far more general than the summarized approach because it solely relies on very basic topological assumptions.

We want to emphasize the importance of the topology of $B$ in Theorem \ref{theorem2}, i.e., the region of the configuration space where the quantum system is confined. The region $B$ represents the geometry of a measuring apparatus (viz., the black box in the terminology of Sec. \ref{Sec3s1}). Had the region $B$ been a disconnected topological space, the idea of scratching, generally speaking, could not have been utilized with the requirement that the classical ensemble should also be confined in $B$.\footnote{
There is, nevertheless, an exception. If $B$ is disconnected, i.e., $B =  B(1) \cup B(2)$ and $B(1) \cap B(2) = \emptyset$, then one might expect that Theorem \ref{theorem2} should be valid without any changes if $\int_{B(1)} \left| \Psi ({\bf q}, t)\right|^2 d^D {\bf q} = C$ and $\int_{B(2)} \left| \Psi ({\bf q}, t)\right|^2 d^D {\bf q} = 1-C$, $\forall t \in [t_i, t_f]$. In this case, we can apply the scratching procedure in each subregion $B(1)$ and $B(2)$ separately since there is no ``exchange'' of the probability between them.
}

The property that the classical ensemble reproduces the momentum distribution is not a primary one since the measurement of any observable can be reduced to measuring coordinates. As a matter of fact, this property can be completely removed from the formulation of Theorem \ref{theorem2} without affecting the rest of the statement; in this case, steps (\ref{GettingRightTangetVector2}) and (\ref{RightMomentumMagnitudeCondition2}) need to be removed from the poof, and the upper bound on $Q$ can be relaxed to $Q > n^{Kn}$. 

Theorem \ref{theorem2} allows for many further generalizations. For example, it would be also interesting to apply Valentini's analysis \cite{Valentini1991a, Valentini1991} to Theorem \ref{theorem2} in future papers.

\acknowledgments
	
The author thanks Robert W. Spekkens and Robert R. Lompay for many vital comments. The Ontario Graduate Scholarship is acknowledged for financial support. 

\appendix 	

\section{A result regarding the Diophantine approximation of a sequence of real numbers with linear constrains}

\begin{lemma}
$\forall \alpha_1^{(1)}, \ldots, \alpha_n^{(1)}, \ldots, \alpha_1^{(K)}, \ldots, \alpha_n^{(K)} \in \mathbb{R}$, such that 
\begin{eqnarray}\nonumber
\sum_{j=1}^n \alpha_j^{(r)} = A^{(r)}/B^{(r)}, \quad A^{(r)}, B^{(r)} \in \mathbb{Z}, \quad r=1,\ldots,K, 
\end{eqnarray}
$\forall Q \in \mathbb{N}$, $Q > \left(n \max_r\left\{ |B^{(r)}|\right\} \right)^{nK}$, $\exists q, a_1^{(1)}, \ldots, a_n^{(1)}, \ldots, a_1^{(K)}, \ldots, a_n^{(K)} \in  \mathbb{Z}$, such that
\begin{eqnarray}
\max_{j=1,\ldots,n \atop r=1,\ldots,K}\left\{ \left| \alpha_j^{(r)} - a_j^{(r)} /q  \right| \right\} < \frac{1}{q Q^{\frac 1 {nK}}}, \quad 0< q \leqslant Q, \label{BoundsOnDiophantineApprox}\\
\sum_{j=1}^n a_j^{(r)}/q = A^{(r)}/B^{(r)}, \quad r=1,\ldots,K. \label{LinearConstrains}
\end{eqnarray}
\end{lemma}
\begin{proof}
According to the Dirichlet approximation theorem \cite{Bernik2002} (see also Ref. \cite{Schmidt1980}), $\forall \alpha_1^{(1)}, \ldots, \alpha_n^{(K)} \in \mathbb{R}$, $\forall Q \in \mathbb{N}$, $Q>1$, $\exists q, a_1^{(1)}, \ldots, a_n^{(K)} \in \mathbb{Z}$, such that
\begin{eqnarray}\nonumber
\max_{j=1,\ldots,n \atop r=1,\ldots,K}\left\{ \left| q\alpha_j^{(r)} - a_j^{(r)}   \right| \right\} < Q^{-\frac 1 {nK}}, \quad 0< q \leqslant Q.
\end{eqnarray}
Therefore, Eq. (\ref{BoundsOnDiophantineApprox}) is proven. To derive Eq. (\ref{LinearConstrains}), consider 
\begin{eqnarray}
&&\left| A^{(r)} q - B^{(r)} \sum_{j=1}^n a_j^{(r)} \right| =
\left|B^{(r)}\right| \left| A^{(r)} q/B^{(r)} -  \sum_{j=1}^n a_j^{(r)} \right| \nonumber\\
&&= \left|B^{(r)}\right| \left| \sum_{j=1}^n \left(q\alpha^{(r)}_j - a^{(r)}_j \right) \right| \leqslant \left|B^{(r)}\right| \sum_{j=1}^n \left| q\alpha^{(r)}_j - a^{(r)}_j \right| \nonumber\\
&&< n  \left|B^{(r)}\right| Q^{-\frac{1}{nK}} \leqslant n\max_r\left\{ |B^{(r)}|\right\} Q^{-\frac{1}{nK}}.
\end{eqnarray}
Assuming that $Q > \left(n \max_r\left\{ |B^{(r)}|\right\} \right)^{nK}$, we obtain 
\begin{eqnarray}\label{AuxiliaryInequality} 
\left| A^{(r)} q - B^{(r)} \sum_{j=1}^n a_j^{(r)} \right| < 1.
\end{eqnarray}
However, since $A^{(r)} q - B^{(r)} \sum_{j=1}^n a_j^{(r)}$ is integer, we conclude that inequality (\ref{AuxiliaryInequality}) is satisfied if and only if Eq. (\ref{LinearConstrains}) takes place. 
\end{proof}

\bibliography{Quantum_vs_Classical}
\end{document}